\documentclass[letterpaper, 10 pt, conference]{ieeeconf}  % Comment this line out if you need a4paper

\IEEEoverridecommandlockouts                              % This command is only needed if 
\usepackage{epsfig} % for postscript graphics files
\usepackage{times} % assumes new font selection scheme installed
\usepackage{amsmath} % assumes amsmath package installed
\usepackage{amssymb}  % assumes amsmath package installed

\usepackage{tikz}
\usetikzlibrary{shapes.geometric, arrows}
\usepackage{color}
\usepackage[ruled]{algorithm2e}
\usepackage[noend]{algpseudocode}

\usepackage{subcaption}
\usepackage{float}

\usepackage{graphicx}
\usepackage{float}
\usepackage{afterpage}
\usepackage[acronym,smallcaps]{glossaries}

\usepackage[english]{babel}
\newtheorem{theorem}{Theorem}

\newtheorem{lemma}[theorem]{Lemma}
\newtheorem{proposition}[theorem]{Proposition}
\newtheorem{definition}{Definition}
\newtheorem{remark}{Remark}

\newtheorem{assumption}{Assumption}

\newcommand{\mbb}[1]{\mathbb #1}
\newcommand{\mbf}[1]{\mathbf #1}
\newcommand{\mcl}[1]{\mathcal #1}

\newcommand{\R}{\mbb{R}}

\newcommand{\Rnx}{\mbb{R}^{n_x}}
\newcommand{\Znv}{\mbb{Z}^{n_v}}
\newcommand{\Rnu}{\mbb{R}^{n_u}}
\DeclareMathOperator {\conv}{conv}
\DeclareMathOperator {\D}{d}

\newcommand{\eDef}{\hfill \Box}
\title{%\LARGE \bf
Turnpike Properties in Discrete-Time  Mixed-Integer Optimal Control 
}

\author{Timm Faulwasser$^{\star}$, Alexander Murray$^{\diamond}$  % <-this % stops a space
	\thanks{$^{\star}$ Timm Faulwasser is with the Institute for Energy Systems, Energy Efficiency and Energy Economics at TU Dortmund University, Germany. {\tt\small timm.faulwasser@ieee.org}}
	\thanks{$^{\diamond}$ Alexander Murray is with the Institute for Automation and Applied Informatics at the Karlsruhe Institute of Technology, Germany. {\tt\small alexander.murray@kit.edu}}%
}

\begin{document}

\maketitle
\thispagestyle{empty}
\pagestyle{empty}

%%%%%%%%%%%%%%%%%%%%%%%%%%%%%%%%%%%%%%%%%%%%%%%%%%%%%%%%%%%%%%%%%%%%%%%%%%%%%%%%
\begin{abstract}
 This note discusses properties of parametric discrete-time Mixed-Integer Optimal Control Problems (MIOCPs) as they often arise in mixed-integer NMPC. We argue that in want for a handle on similarity properties of parametric MIOCPs the classical turnpike notion from optimal control is helpful. We provide sufficient turnpike conditions based on a dissipativity notion of MIOCPs, and we show that the turnpike property allows specific and accurate guesses for the integer-valued controls. More\-over, we show how the turnpike property can be used to derive efficient node-weighted branch-and-bound schemes tailored to parametric MIOCPs. We draw upon numerical examples to illustrate our findings. 

\end{abstract}
%\keywords{ Dissipativity, mixed integer optimal control, hybrid systems}

\section{Introduction}
Recently, optimization-based control of dynamic systems has seen tre\-mendous progress  spanning from industrial applications of   Non-linear Model Predictive Control (NMPC)  \cite{Qin00,Engell12,Rawlings17} to efficient solution methods for Mixed-Integer Opti\-mal Control Problems (MIOCPs) \cite{Belotti13,Gerdts12a,Sager12a,Bemporad_2018,Geyer_2009,Dua02a}. 
By now, for consider\-ably non-linear and non-convex cases imple\-mentation within the milli- to micro-second range can be achieved---provided the underlying continuous OCP can be approximated by a Non-Linear Programm (NLP), see e.g. \cite{Houska11a}. These solution times are possible since the repeatedly solved optimiza\-tion is \textit{parametric} in the initial condition of the underlying dynamic system. In turn, this allows leveraging classical sensitivity properties %\cite{Fiacco76} 
of NLPs, see \cite{Diehl01a}.

While transferring sensitivity concepts from OCPs and NLPs to MIOCPs and MINLPs is not straightforward, convincing cases have been made that also MIOCPs can be solved efficiently, e.g. \cite{Bock_2018,Sager12a,Kirches10b} or \cite[Chap. 7]{Belotti13}. 
These works typically rely on efficient solutions to relaxed auxilliary problems, i.e. often they employ outer con\-vexi\-fica\-tion. However, they do not explicitly exploit the parametric nature of the underlying MIOCP. Indeed only limited results on the analysis of  parametric mixed-integer programs seem to be available.  Multi-parametric MILPs, MIQPs and MINLPs are discussed in \cite{Dua00a,Dua02a,Gueddar12a}. However, these results do not touch upon parametric MIOCPs.

In the present note we leverage dissipativity concepts to explicitly characterize helpful properties of parametric  discrete-time MIOCPs. 
Specifically, we discuss cases of {MIOCP}s exhibiting the so-called \textit{turnpike property}. Turn\-pikes occur in parametric {OCP}s where---for varying initial conditions and increasing horizon length---the time that optimal solutions spend outside of any $\varepsilon$-neighbor\-hood of the optimal steady state is bounded independent of the actual horizon length. The turnpike notion was coined by \cite{Dorfman58} in the 1950s, and early observations of the phenomenon can be traced to works of John von Neumann in the 1930s/1940s \cite{vonNeumann38}. There has been long\-standing interest in turnpikes properties in the context of optimal control in economics \cite{Carlson_1991}.
Recently, there has been renewed interested in turnpikes for continuous OCPs \cite{Trelat_2015,Faulwasser_2017}.

Herein, we argue that in want for a handle on similarity properties of para\-metric MIOCPs all is not lost, and that the turnpike notion enables such a characteriza\-tion. Indeed, we advocate that in-depth investigation of turnpikes in parametric MIOCPs might open new avenues to tailored solution schemes for sequences of MIOCPs, e.g. arising in mixed-integer NMPC. 
Yet, to the best of the authors' knowledge,  there is no analysis of the turnpike phenomenon in MIOCPs available in the literature. To this end, we present a sufficient condition for turnpikes to arise which in turn is based on a dissipativity characterization of MIOCPs and we present a sufficient condition certifying dissipativity of linear-quadratic MIOCPS.  Based on this, we sketch a tailored node-weighted branch-and-bound scheme which leverages the underlying turnpike for the sake of efficient solution. Numerical examples illustrate the benefits of the proposed scheme.

%%%%%%%%%%%%%%%%%%%%%%%%%%%%%%%%%%%%%%%%%%%%
\section{Turnpikes in MIOCPs}\label{sec:prob}
%\subsection{Problem Setting}
We consider MIOCPs of the following form
\begin{subequations} \label{eq:MIOCP}
	\begin{align} 
%		V_N(x_0)\doteq 
		\min_{x(\cdot),u(\cdot),v(\cdot)}~& \sum\limits_{k=0}^{N-1} \ell(x(k),u(k),v(k)) +V_f(x(N))\\
		\text{s.t. } \forall k&\in \{0,\dots,N-1\},\nonumber\\
		x(k+1) &= f(x(k), u(k), v(k)),~x(0)=x_0 \in\mbb{X}_0\\
%		 x(0)&=x_0 \in\mbb{X}_0, \\% + \begin{pmatrix}B_1 & B_2\end{pmatrix}\begin{pmatrix}u_k\\ z_k\end{pmatrix}\\
		x(k)&\in  \mathbb{X}\subseteq\Rnx, \quad u(k)\in  \mathbb{U}\subseteq\Rnu\\
		v(k)&\in  \mathbb{V}\subseteq\Znv, \label{eq:MIOCP_Z}
	\end{align}
\end{subequations}
where the stage cost $\ell:\Rnx\times\Rnu\times\Znv \to \R$, the terminal cost $V_f:\Rnx\times\Rnu\times\Znv \to \R$, and the dynamics $f:\Rnx\times\Rnu\times\Znv \to \Rnx$ are assumed to be Lipschitz continuous in $x$, $u$, and $v$. The constraint sets $\mbb{X}$, $\mbb{U}$ and $\mbb{V}$ are assumed to be compact. The core challenge in the discrete-time OCP \eqref{eq:MIOCP}  is that the input $v$ is assumed to take only discrete values, cf. \eqref{eq:MIOCP_Z}.

Optimal solutions, provided they exist, are written as % $x^\star(\cdot; x_0),u^\star(\cdot;x_0)$, $v^\star(\cdot;x_0)$. Occasionally, we use the shorthand
$$
z^\star(\cdot; x^0) =\begin{bmatrix} x^\star(\cdot;x_0) & u^\star(\cdot;x_0) &v^\star(\cdot;x_0) \end{bmatrix}^\top
$$
denoting the optimal primal triplet. Whenever no confusion can arise, we suppress the dependence on the initial condition $x_0$. 
The stationary counterpart of \eqref{eq:MIOCP} is the MINLP
\begin{subequations} \label{eq:MISOP}
	\begin{align} 
	\min_{\bar x, \bar u, \bar v}&~  \ell(\bar x,\bar u,\bar v)\\
	\text{s. t. } &\nonumber\\
	\bar x &= f(\bar x, \bar u, \bar v), \\% + \begin{pmatrix}B_1 & B_2\end{pmatrix}\begin{pmatrix}u_k\\ z_k\end{pmatrix}\\
	\bar x&\in  \mathbb{X}\subseteq\Rnx, \quad \bar u\in  \mathbb{U}\subseteq\Rnu\\
	\bar v&\in  \mathbb{V}\subseteq\Znv. \label{eq:MISOP_Z}
	\end{align}
\end{subequations}
Simliar to before, the optimal solution is denoted by  $\bar z^\star = \begin{bmatrix} \bar x^\star& \bar u^\star& \bar v^\star \end{bmatrix}^\top$.
We are interested in studying the similarity properties of solutions to \eqref{eq:MIOCP} for varying initial conditions $x_0\subset \mbb{X}_0 \subseteq \mbb{X}$ and varying horizon lengths $N\in \mbb{N}$. 
Put differently, we are interested in analyzing MIOCP \eqref{eq:MIOCP} as a problem parametric in $x_0$ and $N\in\mbb{N}$.

\subsubsection*{Illustrative Example}

%\begin{example}
	We consider a straight-forward modification of a simple problem presented in \cite{Faulwasser_2018,Gruene_2013} which reads
	\begin{align}  \label{simple_example}
	\min_{x(\cdot),u(\cdot),v(\cdot)}& \sum\limits_{k=0}^{N-1} u(k)^2+{\textstyle\frac{1}{2}}v(k)^2\nonumber\\
	\text{s. t. }	\forall k&\in \{0,\dots,N-1\},\\
	x(k+1) &= 2x(k) + u(k) + v(k)-1, \quad x(0) = x_0 \nonumber\\
	 \begin{bmatrix}x(k)~u(k)~v(k)\end{bmatrix}^\top &\in [-2,2]\times [-3,3] \times \{-1,0,1\}.  \nonumber
	\end{align}
%\end{example}
Figure \ref{fig:example_plot} shows the results for the horizon $N=30$ and several initial conditions $x_0$.
Note that $z^\star(\cdot;x_0)$ is  close to its optimal steady state value $\bar z^\star =[1~0~0]^\top$ for a large part of the time horizon. This similarity of optimal solutions for different initial conditions is called \textit{turnpike phenomenon}; a definition for {MIOCP}s will be given below. 

\begin{figure}[t]
		\centering
	\includegraphics[width=0.4\textwidth]{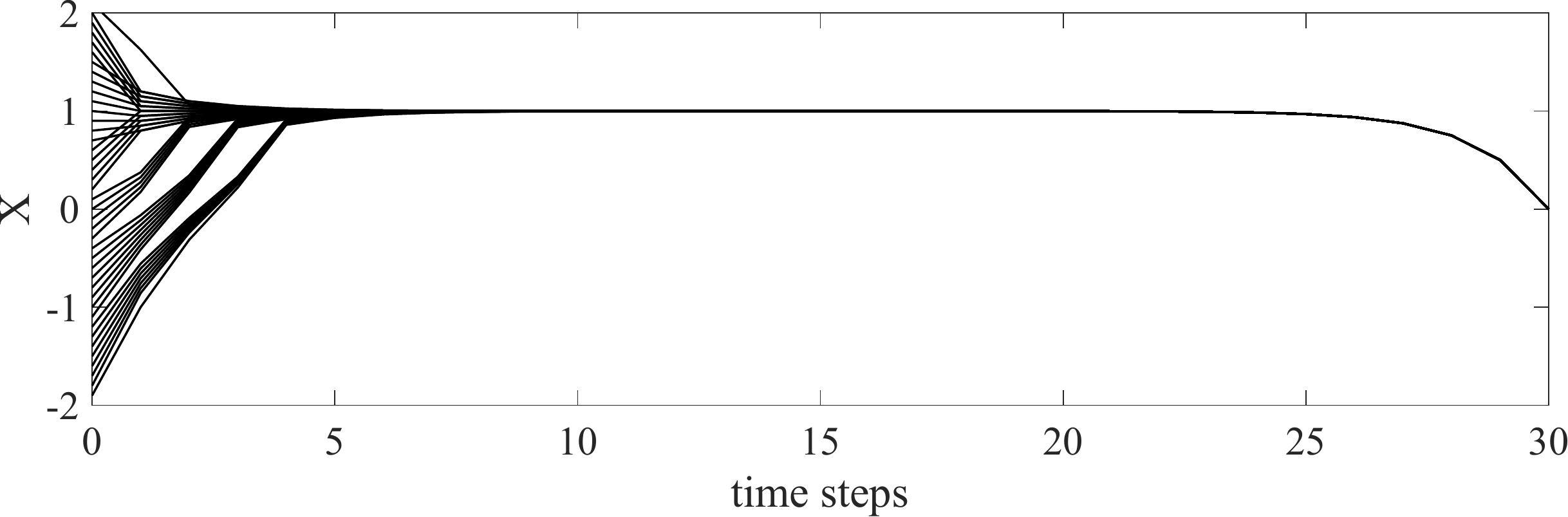}
	\includegraphics[width=0.4\textwidth]{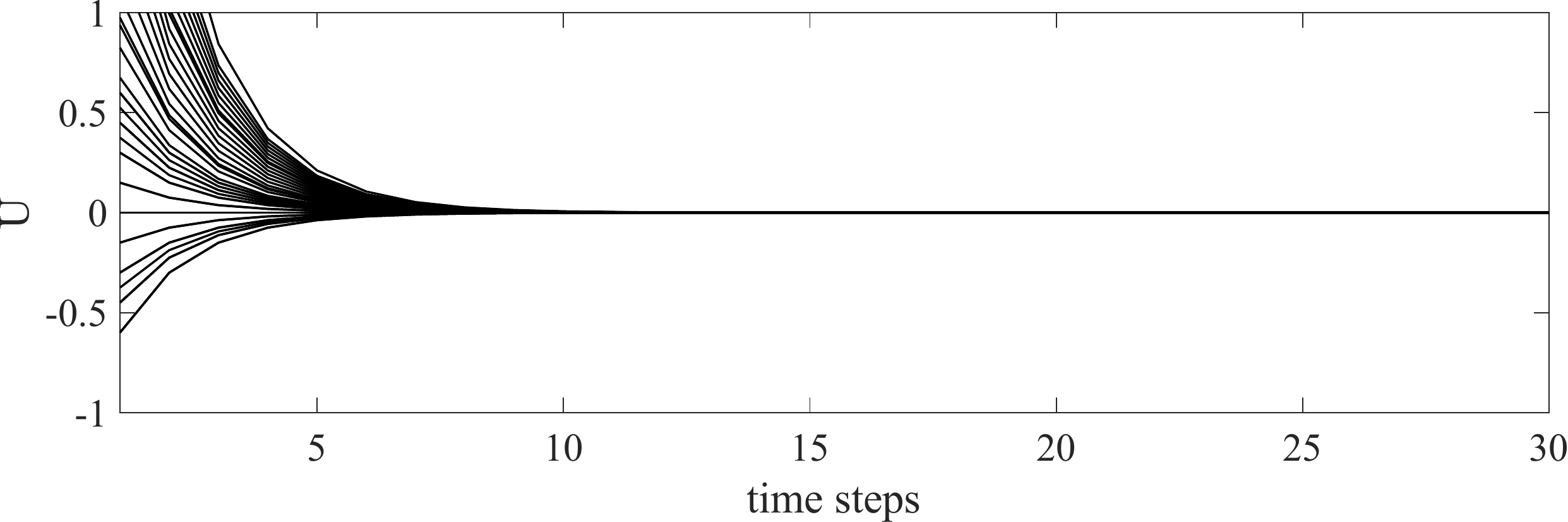}
	\includegraphics[width=0.4\textwidth]{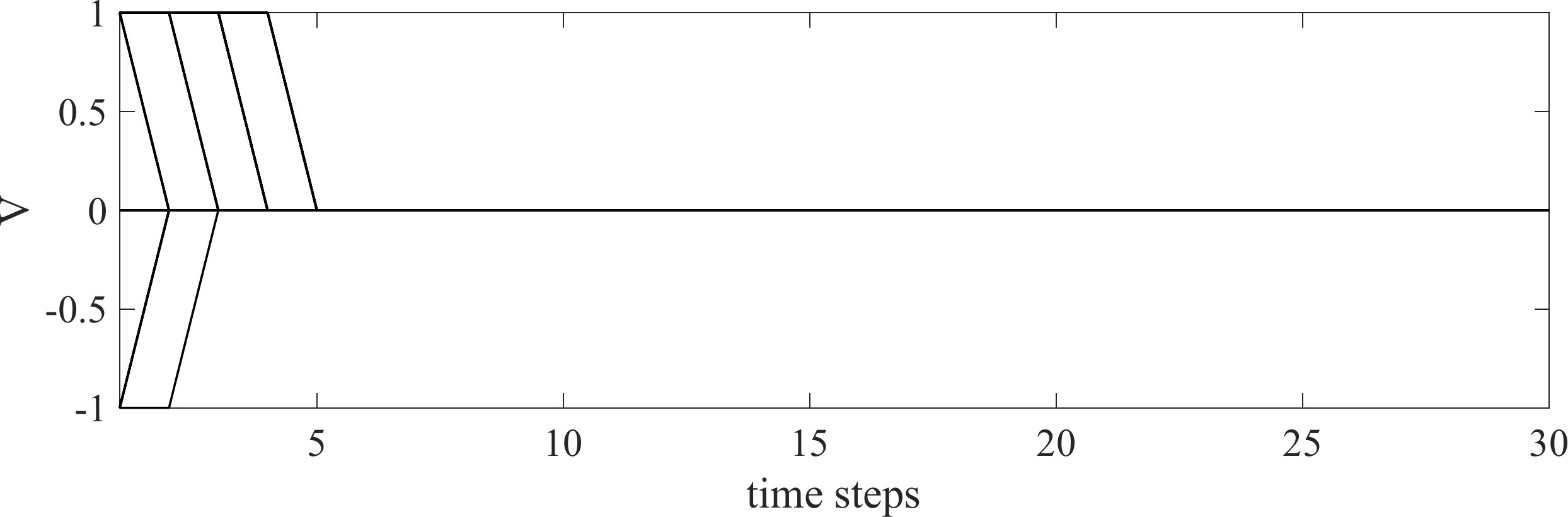}
	\caption{Example plot of the optimal solutions for Problem \eqref{simple_example} with $N=30$ and a variety of initial copnditions $x_0$.}
	\label{fig:example_plot}
\end{figure}

\subsection*{Sufficient Conditions for Turnpikes in MIOCPs}
First, we suggest a rigorous definition of the phenomenon:
\begin{definition}[Mixed-integer turnpike property]\label{def:turnpike}
	MIOCP \eqref{eq:MIOCP} is said to have an input-state turnpike if for all $x_0 \in \mbb{X}_0$ and all $N\in \mbb{N}$
	$$
	\#\mathbb{Q}_\varepsilon  \geq N-\frac{C}{\alpha(\varepsilon)}
	$$
	holds, where 
	\begin{equation}\label{eq:SetQ}
	\mathbb{Q}_\varepsilon:=\{k\in\{0,\dots,N-1\}\hspace*{1mm}|\hspace{1mm} ||(z^*(k;x_0))-\bar z||\leq \varepsilon \},
	\end{equation}
	$\#\mbb{Q}_\varepsilon$ is the cardinality of $\mbb{Q}_\varepsilon$,
	and $\alpha \in \mathcal{K}_\infty$.\footnote{$ \mathcal{K}_\infty$ refers to the set of continuous functions $\mbb{R}_0^+\to \mbb{R}_0^+$ which are $0$ at $0$, strictly monotonously increasing, and satisfy $\displaystyle \lim_{s\to \infty}\alpha(s) = \infty$.}% is as in Definition \ref{def:dissip}.
	$\eDef$
\end{definition}
The above definition is a straight-forward extension of continuous concepts \cite{Gruene_2013,Faulwasser_2017}.\footnote{Indeed the setting of \cite{Gruene_2013} applies to generic normed spaces, hence it includes the setting considered here.}
Its basic meaning is that the amount of time an optimal triplet $z^\star(\cdot)$ spends inside of an $\varepsilon$-ball centered at $\bar z$ is at least $N-\frac{C}{\alpha(\varepsilon)}$, or alternatively, for all horizons the amount of time spent outside of the
$\varepsilon$-ball is bounded independently of $N$ by $\frac{C}{\alpha(\varepsilon)}$. 
For the sake of brevity, we refer to $\bar z = \bar z^\star$ simply as \textit{the turnpike}. We remark that the first part of the optimal solution approaching the turnpike is often called the \textit{entry arc}, while the distinctive departure from the turnpike is called the \textit{leaving arc}. Note that a leaving arc does not need to occur. 

To the end of certifying turnpikes in mixed-integer problems, we recall and adapt  an established dissipativity notion for {OCP}s to the mixed-integer setting.
\begin{definition}[Strict dissipativity with respect to $\bar z$]
	\label{def:dissip}~\\
	A system 
	$x(k+1) = f(x(k), u(k), v(k))\doteq f(z(k))$
	 is said to be \textit{dissipative with respect to a steady-state tuple} $\bar z \in \mathbb{X}\times \mathbb{U} \times \mathbb{V}$, if there exists a bounded function $\lambda: \mathbb{X}\rightarrow\mathbb{R}$ such that for all $z= \begin{bmatrix} x & u & v\end{bmatrix}^\top \in\mathbb{X}\times\mathbb{U}\times\mathbb{V}$
	\begin{subequations}
		\begin{equation}\label{eq:dissip1}
		\lambda(f(z))-\lambda(x)\leq\ell(z)-\ell(\bar z),
		\end{equation}
		with $\ell$ from \eqref{eq:MIOCP}.

		If, additionally, there exists $\alpha\in \mathcal{K}_\infty$ such that
		\begin{equation}\label{eq:dissip2}
		\lambda(f(z))-\lambda(x)\leq -\alpha(\|z-\bar z\|)+\ell(z)-\ell(\bar z),%\nonumber
		\end{equation}
	\end{subequations}
	then the system is said to be \textit{strictly dissipative with respect to} $\bar z = \begin{bmatrix} \bar x & \bar  u & \bar v\end{bmatrix}^\top$.
	
	Moreover, if the above dissipativity notions holds solely along optimal solutions of \eqref{eq:MIOCP} with $x_0 \in \mbb{X}_0$, then {MIOCP} \eqref{eq:MIOCP} is said to be strictly dissipative with respect to $\bar z$. $\eDef$
\end{definition}
The following property follows directly from the above definition. 
\begin{lemma}[Optimality of $\bar z$]
Let MIOCP\eqref{eq:MIOCP} be strictly dissipative with respect to $\bar z$ in the sense of Definition \ref{def:dissip}, then $\bar z=\bar z^\star$ is the unique globally optimal solution to \eqref{eq:MISOP}.$\eDef$
\end{lemma}

The next assumption helps to establish existence of turnpikes.

\begin{assumption}[Exponential reachability of $\bar x^\star$]	\label{ass:reachable}	
	For all $x_0\in\mathbb{X}_0$, there exist infinite-horizon admissible inputs $u_\infty, v_\infty:\mbb{N}\cap\infty\to \mbb{U}\times\mbb{V}$ and constants $c\geq 0, \rho \in [0,1)$ such that 
$
	\|x(k; x_0, u_\infty(\cdot), v_\infty(\cdot)) - \bar x^\star\| \leq c \rho^k,
$
	i.e. the optimal steady state $\bar x^\star$ is exponentially reachable. $\eDef$
\end{assumption}
\begin{proposition}[{Turnpikes in {MIOCP}s}]
	\label{prop:turnpike}	
	Let Assumption \ref{ass:reachable} hold and suppose that, for all $x_0 \in \mbb{X}_0$, the MIOCP \eqref{eq:MIOCP} is strictly dissipative with respect to $\bar z^\star$. Then the MIOCP \eqref{eq:MIOCP} has an input-state turnpike at $\bar z^\star = \bar z$. $\eDef$ 
\end{proposition}
The proof follows the pattern of the ones presented in  \cite{Faulwasser_2018,Gruene_2013} for continuous discrete-time OCPs and is thus omitted.

It is worth noticing that the turnpike property  implies the following  property of optimal solutions:
\begin{proposition}[Integer controls exactly at turnpike] \label{prop:vstar}
	For all $x_0\in\mbb{X}_0$, let MIOCP \eqref{eq:MIOCP} have an input-state turnpike at $\bar z^\star = \bar z$ in the sense of Definition \ref{def:turnpike}. Then, for all $\varepsilon \in (0, 1)$ and sufficiently large $N\in\mbb{N}$, the optimal integer controls $v^\star(\cdot;x_0)$ satisfy
	\[
	\hspace*{1.95cm}v^\star(k;x_0) \equiv \bar v^\star \text{ for all }k\in\mbb{Q}_\varepsilon.
	\hspace*{2.05cm}\eDef
	\]
\end{proposition}
\begin{proof}
Recall that $\|z^\star(k;x_0) -\bar z^\star\|\geq \|v^\star(k;x_0) -\bar v^\star\|$, and note the fact that for the integer controls $\|v^\star(k;x_0) -\bar v^\star\|<1 \Leftrightarrow \|v^\star(k;x_0) -\bar v^\star\| =0$. Combining both and taking the definition of  $\mbb{Q}_\varepsilon$ in \eqref{eq:SetQ} into account yields the assertion. 
\end{proof}

It is fair to ask how the dissipativity property from Definition \ref{def:dissip} can be numerically verified for MIOCPs. 
The next result provides a sufficient condition. 
\begin{theorem}[Dissipativity of linear-quadratic MIOCPs]\label{thm:DIchk}
Consider MIOCP \eqref{eq:MIOCP} and let the dynamics and the stage cost be given by   
\begin{align*}
x^+ &= Ax + B_1u +B_2v,\\
\ell(x,u,v) &= x^\top Q x +\begin{bmatrix}
u \\v\end{bmatrix}^\top R \begin{bmatrix}
u \\v\end{bmatrix}+ q^\top x +r^\top\begin{bmatrix}
 u\\v\end{bmatrix},
\end{align*}
with potentially indefinite matrices $Q$ and $R$.
If there exists a matrix $P\in\R^{n_x\times n_x}$ such that
\begin{equation} \label{eq:DI_certificate}
Q +P - A^\top P A \succ 0,
\end{equation}
then MIOCP \eqref{eq:MIOCP} is strictly dissipative, and there exists $p\in\Rnx$
such that $x^\top P x+p^\top x$ is a corresponding storage function.$\eDef$
\end{theorem}
\begin{proof}
The proof combines the available storage characterization of dissipativity with recent results from \cite{Gruene18a}. 

The available storage is given by 
\begin{align*}
S^a_\mbb{V}(x) =& \sup_{z(\cdot), N} ~\sum_{k=0}^N -\alpha(\|(z_k) - \bar z\|)+ \ell(z_k) + \ell(\bar z)\\
\text{subject to }&\\
x(k+1) &= Ax(k) +\begin{bmatrix}B_1&B_2\end{bmatrix}\begin{bmatrix}u(k)\\v(k)\end{bmatrix}, ~ x(0)=x\\
x(k)&\in\mbb{X}, ~ u(k)\in\mbb{U}, ~v(k)\in \mbb{V},
\end{align*}
whereby the subscript $\cdot_\mbb{V}$ refers to the constraints on the discrete control input. 
Indeed the MICOP \eqref{eq:MIOCP} is dissipative on $\mbb{X}_0$ if $S^a_\mbb{V}(x) < \infty$ for all $x \in \mbb{X}_0$, cf.  \cite{Willems72a}.

Step 1: Let $\conv\mbb{V}$ be the convex hull of $\mbb{V}$, which is compact if and only if $\mbb{V}\subset\mbb{Z}$ is compact. The inclusion $\mbb{V}\subset \conv\mbb{V}$ implies that 
\[
S^a_\mbb{V}(x) \leq S^a_{\conv{\mbb{V}}}(x),
\]
i.e. dissipativity of the continuously relaxed OCP certifies dissipativity of the 
MIOCP. 

Step 2: For the continuously relaxed problem it has been shown in \cite[Lem. 4.1]{Gruene18a} that \eqref{eq:DI_certificate} is a necessary and sufficient condition for strict dissipativity with quadratic storage function on bounded subsets of $\Rnx$. Combining both facts yields the assertion. 
\end{proof}

\begin{remark}[Turnpikes and parametric MIOCPs]~\\
The above results highlight that \textit{turnpikes are to be understood as similarity properties of parametric MIOCPs}. We remark that the Propositions \ref{prop:turnpike}  and \ref{prop:vstar} apply to non-convex \textit{and} convex MIOCPs alike.\footnote{Similar to \cite{Bonami2012,Lubin2017}, we say an {MINLP}/MICOP is convex if its continuous relaxation is convex.} They also allow for linear and nonlinear dynamics. Moreover, Theorem \ref{thm:DIchk} addresses linear-quadratic problems, which however also do not need to be convex, i.e. indefinite matrices $Q,R$ are included. We remark that the solution $P$ to the Lyapunov equation \eqref{eq:DI_certificate} does not need to be positive definite. Moreover, notice that dissipativity of the MICOP does not depend on the specific choices of the linear weightings $q, r$ in the cost function.
$\eDef$
\end{remark}

 We proceed to sketch how these findings can be leveraged to design efficient solution strategies for {MIOCP}s. 

%%%%%%%%%%%%%%%%%%%%%%%%%%%%%%%%%%%%%%%%%%%%
\section{Branch and Bound with Node Weighting}\label{sec:heur}
Here we focus on an intuitive strategy for branch-and-bound algorithms that allows for exploiting a-priori knowledge of a turnpike. 
Specifically we suggest prioritizing exploration of nodes of the decision tree which are associated with the steady state optimal integer decisions at the turnpike (which is usually in the middle of the time horizon).\footnote{Indeed leveraging the concept of \textit{exact turnpikes}~\cite{kit:faulwasser17a}, one can extend Proposition \ref{prop:vstar} and show that the integer-valued controls will be exactly at the turnpike in the middle part of the horizon. Note that, without additional assumptions, in the continuous setting the solutions merely stay \textit{close} to the turnpike but do not need to reach it \textit{exactly}.} 
\begin{figure}[t] 
	%\footnotesize
	\begin{center}
		\begin{minipage}{0.49\textwidth}
			\rule{1\textwidth}{0.3mm}\\
			\textbf{%\textcolor{red}{Typeset in algorithm environment!} 
			Algorithm~1: Branch and Bound with Node Weighting}\\
			\rule{1\textwidth}{0.3mm}\\
			\textbf{Input:} Guesses $\mcl{V}_0 =\{\mbf{V}_{0, 1}, \dots, \mbf{V}_{0,M}\}$ and corresponding weights $\mcl{W}_0=\{w_{0,1}, \dots, w_{0,M}\}$. Termination tolerance $\epsilon>0$. Default search strategy (depth-first, breadth-first, $\dots$) and corresponding weights $\mcl{W}$.\\
			\textbf{Preparation:} Set $U=\infty$, $L=-\infty$, $\mcl{T}=\emptyset$. Re-index nodes $\mcl{N}$ according to weights $\mcl{W}_0$. Candidate node set $\mcl{S}=\{0\}$.\\
			\textbf{\textbf{While} $\mcl{S} \neq \emptyset$:} %\textcolor{red}{We need a better notation for $S$}
			\\[-0.3cm]
			\begin{enumerate}
				\setlength{\itemsep}{2pt}
				
				\item \label{step::2}  
				$\displaystyle \forall n \in \mcl{S}$\quad \\
				$\displaystyle \tilde w(n)=w(n)+\sum_{i=1}^{M} w_{0,i} \left(\D_\mbb{V}(\mbf{V}_{0,i})-\|\mbf{V}_{0,i}-\mbf{V}_n\|_0\right)$
				
%				\textcolor{red}{What does $\dim(v_0)$ mean here? Number of chosen elements defined by $v_0$?} 
				%where $w$ is the weighting from the default search strategy and $M(v_0)=dim(v_0)$.
				\item \label{step::3} $\displaystyle \tilde n = \underset{n\in\mcl{S}}{\arg\max} ~ \tilde w(n)$ and $\mcl{S}\leftarrow \mcl{S}\setminus \{\tilde n\}$
%				\item \label{step::3} $\displaystyle \mcl{T} = \max_{n\in\mcl{S}} \tilde w(n) \cup \mcl{T}$
				
			%	Choose the node with the highest weighting computed in Step \ref{step::2} and add it to $T$.
				
				\item \label{step::nlp} \textbf{Solve} NLP($\mbf{V}_{\tilde n}$) for $J^\star(\mbf{V}_{\tilde n})$ and 
				${z}^\star(\mbf{V}_{\tilde n})$ and $\mcl{T}\leftarrow \mcl{T}\cup \tilde n$ %Problem \eqref{eq:MIOCP} such that $v(k)\in conv(\mathbb{V})$ if $n(k)=\emptyset$ and $v(k)=n(k)$ otherwise.
				
				\item \label{step:ub} \textbf{If} ${z}^\star(\mbf{V}_{\tilde n})$ is feasible in MIOCP \eqref{eq:MIOCP} and $J^\star(\mbf{V}_{\tilde n})<U$, then $U \leftarrow J^\star(\mbf{V}_{\tilde n})$, proceed to Step \ref{step:term}.\\				
				\textbf{If} $J^\star(\mbf{V}_{\tilde n})>U$ proceed to Step \ref{step::2}. \\
				\textbf{Else} add the child nodes of $\tilde n$ to $\mcl{S}$,  proceed to Step \ref{step:term}.
				
				\item \label{step:term} 
				$	L \leftarrow  \displaystyle \min_{n\in\mcl{P}(\mcl{S})}\{J^\star(\mbf{V}_{n})\}$\\
				%\textcolor{red}{this would only ever yield the value of the root node as the lower bound. The actual function is the lowest value of all leaf nodes of T}\\
				\textbf{If} $U-L\leq \epsilon$ terminate. \\
				\textbf{Else} proceed to Step \ref{step::2}.

			%	Compute the lower bound $L$ based on the values of the nodes stored in $T$. If $U-L\leq \varepsilon$ then terminate. Otherwise, proceed to Step \ref{step::2}.
				
			\end{enumerate}
			
			\removelastskip\rule{1\textwidth}{0.3mm}
		\end{minipage}
	\end{center}
\end{figure}

To this end, let $\mbf{V} \in \mbb{R}^{n_v\times N}$ be a full or partial guess of the sequence of optimal integer decisions $v^\star(\cdot;x_0)$ for Problem \eqref{eq:MIOCP}, and let $\mbf{V}(k)$ denote its $k$-th column, which corresponds to the integer decisions for the $k^{th}$ time step.

Consider the following NLP relaxation of MIOCP \eqref{eq:MIOCP}
\begin{equation} \label{eq:NLP}
\text{NLP}(\mbf{V}) = \left\{\text{\eqref{eq:MIOCP} with } \left\{
\begin{array}{l l}
v(k) = \mbf{V}(k) & \text{if }  \mbf{V}(k) \in \mbb{V}\\
v(k) \in \conv\mbb{V}&\text{else} \end{array}\right.\right\},
\end{equation}
i.e. $\eqref{eq:MIOCP_Z}$ is continuously relaxed.
This relaxed problem takes $\mbf{V}$ as a (partial or full) guess of the integer variables $v(0), \dots, v(N-1)$. The condition  $\mbf{V}(k) \not\in \mbb{V}$ implicitly encodes the relaxation   $\mbf{V}(k) \in \conv\mbb{V}$.
We write  $J^\star(\mbf{V})$ and ${z}^\star(\mbf{V})$ to denote the optimal performance bound, respectively, the (partially) relaxed solution triplet obtained from solving \text{NLP}($\mbf{V}$). 
Let $\D_\mbb{V}(\mbf{V})$ denote the number of feasible integer components of $\mbf{V}$, i..e.  $\D_\mbb{V}(\mbf{V}) = N\cdot n_v$ means for all $k \in \{0, \dots, N\}$ $\mbf{V}(k) \in \mbb{V}$, and  $\D_\mbb{V}(\mbf{V}) = 0$ implies that all integer decisions are relaxed in \eqref{eq:NLP}. 

Algorithm~1 lays out a straight-forward method for making use of $\mbf{V}$ in the branch-and-bound decision tree that goes beyond warm-starting. For a decision tree consisting of nodes $\mcl{N}$, let $\mbf{V}_n$ denote the integer decision related to the node $n$ of the branch-and-bound tree, and let $\mcl{P}(\mcl{S})$ denote the parent nodes for a node set $\mcl{S}\subseteq \mcl{N}$.
In Step \ref{step::2} of Algorithm~1, the task is to determine which node should be explored first. To this end, the given guesses $\mcl{V}_0$ and their corresponding chosen weights $\mcl{W}_0$ are used. The number of different elements between the guess $\mbf{V}_{0,i}$ and the integer vector at node $n$ of the branch-and-bound tree ,$\mbf{V}_{n}$, is used to compute the weight update.\footnote{We remark that slight abuse of notation is caused by suppressing, for the sake of readability, the vectorization in $\|\mbf{V}_{0,i}-\mbf{V}_n\|_0$ in Step \ref{step::2}.}  Step \ref{step::3} passes the node $\tilde n$ with the highest weighting and keeps track of explored nodes in the branch and bound tree, which is needed for computing the lower bound in Step \ref{step:term}. Step \ref{step::nlp} fixes some integer variables according to the chosen node $\tilde n$, and the others take values from the convex relaxation of $\mathbb{V}$, cf. \eqref{eq:NLP}.  Step \ref{step:ub} determines whether the upper bound should be updated, and if not whether child nodes of $n$ should be added to the list of candidate nodes. Finally, Step \ref{step:term} checks whether the algorithm should terminate due to the upper and lower bounds being within the given tolerance $\epsilon$.

It is easy to see that Algorithm~1 inherits properties of standard branch-and-bound schemes, i.e. if the relaxed {NLP}s in Step \ref{step::nlp} are solved to global optimality then a globally optimal solution to \eqref{eq:MIOCP} will be found, as the worst case is full enumeration \cite{Hansen03a}.  Observe that as such Algorithm~1 does not require any knowledge about a turnpike property. However, the turnpike of the underlying {MIOCP} can and should be encoded in ($\mcl{V}_0$,  $\mcl{W}_0$). For example, this can be done by formulating guesses leveraging the insight of Proposition \ref{prop:vstar}. In other words, the initial guesses  ($\mcl{V}_0$, $\mcl{W}_0$) shall comprise high priority cases where $v(k) \equiv \bar v^\star$ holds.

\begin{remark}[Alternative branching / weighting strategies]
We remark that one can imagine a generic branching strategy as following a specific node-weighting pattern. For example, a depth-first search weights lower nodes in the decision tree more highly, while a breadth-first search would weight the top nodes more than the bottom ones. However,  in our prototpyical implementation, we observe that it is not advisable to explictly weight large swathes of the branch and bound tree nodes since computing the weighting and storing the resulting data can be computationally expensive. Rather, as illustrated in the example of the next section, just a few nodes should be weighted. The question of how many nodes to weight and how to construct weighting strategies for classes of problems that are not in the form of \eqref{eq:MIOCP} are still open problems. $\eDef$
\end{remark}

%%%%%%%%%%%%%%%%%%%%%%%%%%%%%%%%%%%%%%%%%%%%
\section{Numerical Examples}\label{sec:sim}
To test the performance of Algorithm~1 we consider parametric convex linear-quadratic MIOCPs, which result in MIQPs and for which dissipativity can be checked via Theorem \ref{thm:DIchk}. The solutions are obtained using a prototypcial implementation of Algorithm~1 within MATLAB R2019a relying on {CasADi v3.5.0} \cite{Andersson19a} and {IPOPT} to solve the {QP} subproblems. The imple\-mented branch-and-bound algorithm starts with a depth-first branching strategy to obtain an upper bound and then seeks to improve the lower bound as quickly as possible. All numerical experiments were run on a 2.9GHz Intel Core i5-4460S CPU with 8GB of RAM.

\subsubsection*{Example 1}\label{sec:ex1}
We consider the dynamics proposed in \cite{Bemporad_1999}:
$$
x_1(k+1) = 
\begin{cases}
\phantom{-}0.8x_1(k) + u(k),\quad \text{if } x_1(k) \geq 0,\\
-0.8x_1(k) + u(k), \quad\text{if } x_1(k) < 0, \nonumber\\
\end{cases}
$$
with $\underline{x} \leq x_1(k) \leq \overline{x}$. This can be reformulated into a mixed-integer system of equations through the introduction of the continuous state variable $x_2$ and the discrete input $v$:%\vspace*{-3mm}
\begin{subequations} \label{eq:mi_dyn}
\begin{align} 	
x_1(k+1) &= 0.8x_2(k) + u(k), \\
2v(k)\underline{x}  &\leq x_2(k)+x_1(k) \leq 2v(k) \overline{x}, \\
2(v(k)-1) \overline{x} &\leq x_2(k)-x_1(k) \leq 2(v(k)-1)\underline{x} , \\
(1-v(k))\underline{x}  &\leq x_1(k) \leq v(k) \overline{x}, \\
\underline{u} &\leq u(k) \leq \overline{u},  \\
v(k) &\in \{0,\,1\}.
\end{align}
\end{subequations}
It is easy to verify that if $v(k)=0$ then $x_2(k)=-x_1(k)$ and $x_1(k)<0$, and if $v(k)=1$ then $x_2(k)=x_1(k)$ and $x_1(k) \geq 0$. The considered {MIOCP} reads:
%\begin{subequations} 
\begin{align} \label{eq:lit_ex}
\min_{z(\cdot)}  \sum_{k=0}^{N-2} 
l^\top
%\begin{smallmatrix}\phantom{0}-10\\\phantom{-}100\\\phantom{-0}10\\\phantom{-}100\end{smallmatrix}^\top
\left[\begin{smallmatrix}x_1(k) \\ x_1^2(k)\\ \phantom{_1}u(k) \\ u^2(k)\end{smallmatrix}\right] + 
%\begin{smallmatrix}-1000\\\phantom{-0}100\\\phantom{-00}10\\\phantom{-0}100\end{bmatrix}^\top
l_f^\top
\left[\begin{smallmatrix}x_1(N-1) \\ x_1^2(N-1)\\ \phantom{_1}u(N-1) \\ u^2(N-1)\end{smallmatrix}\right] \\
 \text{subject to } \eqref{eq:mi_dyn} \text{ and } x_1(0)=x_0, \nonumber
\end{align}
with  $l^\top = \begin{bmatrix}-10&100&10&100\end{bmatrix}$ and $l_f^\top = \begin{bmatrix}-1000&100&10&100\end{bmatrix}$.

All results presented use $\underline{x}=-1$, $ \overline{x}=1$, $\underline{u}=0.5$, and $\overline{u}=0.5$.
Algorithm~1 is provided a collection of complete discrete solution guesses and weights, as shown in Table~\ref{t:sol_vecs}.
Figure \ref{fig:results1} depicts the optimal solution with $N=20$ and $x_0 = 1$. Note that one can clearly spot the turnpike at {$\bar z^\star = \begin{bmatrix}-\frac{1}{9}& \frac{1}{9} &-0.2 & 0\end{bmatrix}$}, which corresponds to solving \eqref{eq:MISOP} for this example. 
Shown in Table \ref{t:ex1} are the aggregated results for each combination of $x_1(0)=\{\underline{x},\underline{x}+0.1,\dots,\overline{x}-0.1,\overline{x}\}$
for $N \in\{10,20\}$.
These results compare  Algorithm~1 without initial guesses (``std. B\&B") to Algorithm~1 leveraging the collection of solution vectors and weights from above (``node wthg"). The guesses and weights used are shown in Table \ref{t:sol_vecs}.
\begin{figure}[t]
		\centering
\hspace*{-3mm}	\includegraphics[width=0.5\textwidth]{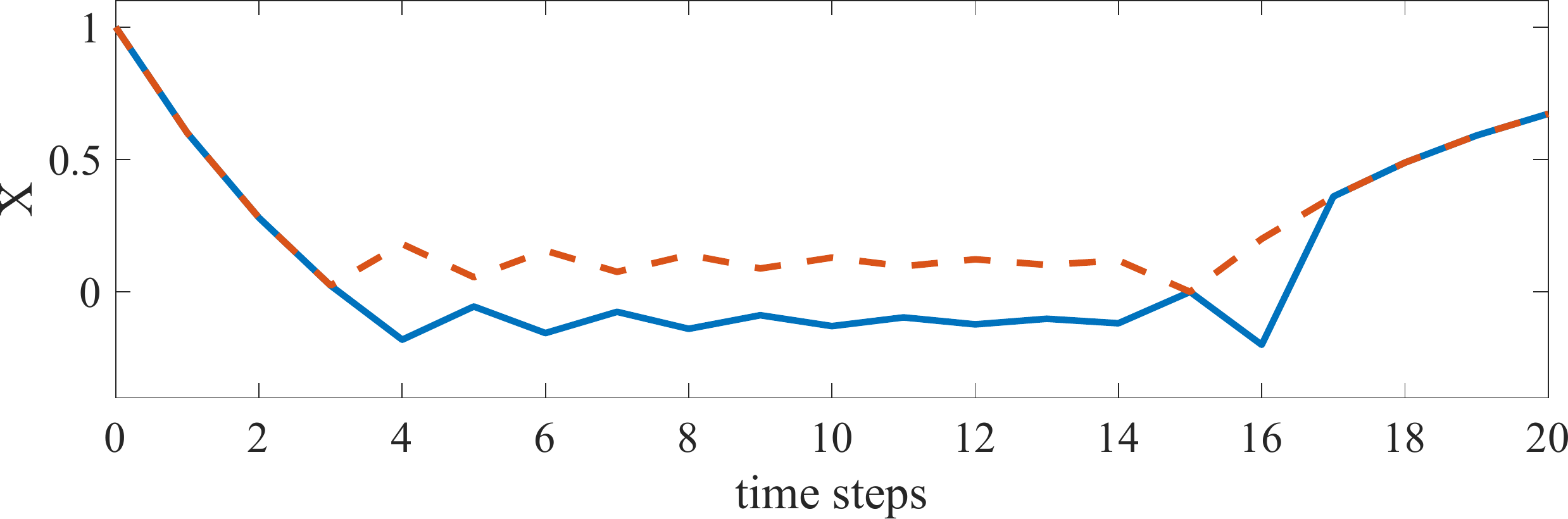}\\
	\includegraphics[width=0.5\textwidth]{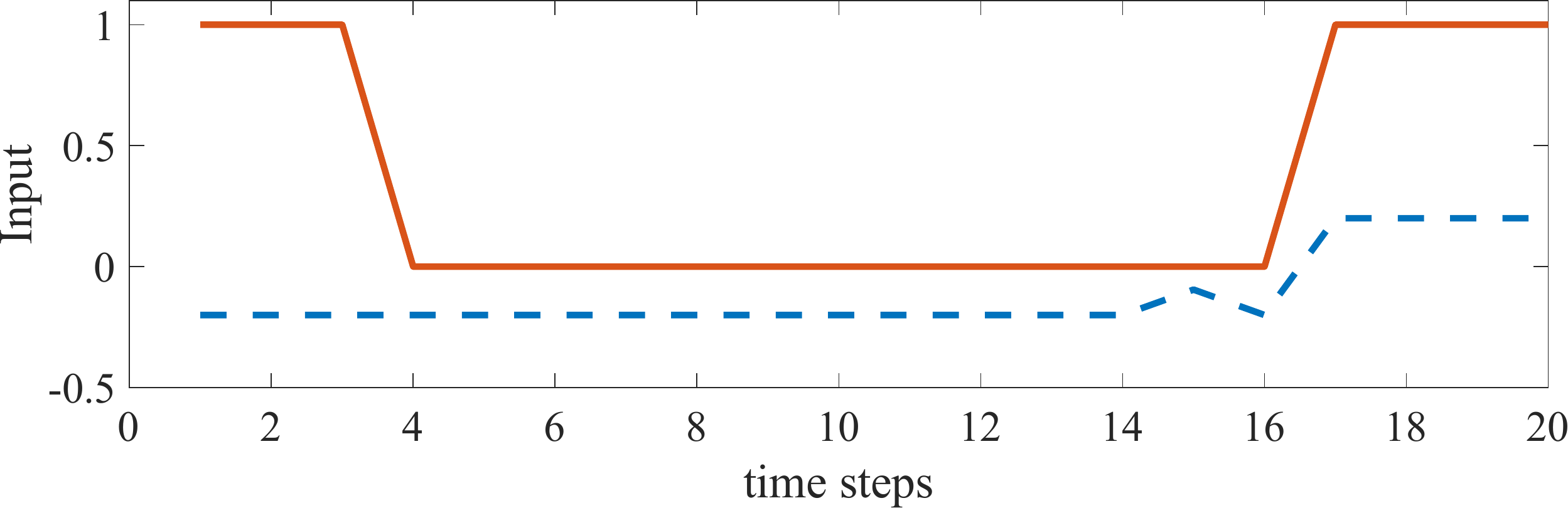}
	\caption{Opt. solutions for Example 1 with $N=20$ and $x_0=1$. Left: state trajectories $x_1$ (dashed blue) $x_2$ (solid red). Right: constrols $u$ (dashed blue) $v$ (solid red).}
	\label{fig:results1}
\end{figure}
\begin{table}[t]
	\centering
	\setlength{\tabcolsep}{10pt}
	\renewcommand{\arraystretch}{1.3}
	\caption{Guesses $\mcl{V}_0$ and weights $\mcl{W}_0$ for Example \eqref{eq:lit_ex}.}
	\label{t:sol_vecs}	
	\begin{tabular}{|c|c|| c|c|}	
		\hline
		$\mbf{V}_{0,i}$ &  $w_{0,i}$ &  $\mbf{V}_{0,i}$ &  $w_{0,i}$\\ 
		\hline
			[1,1,0,$\dots$,0] & 1 &	{[1,1,1,0,$\dots$,0]} & 2	\\
			{[1,1,1,0,$\dots$,0,1]} & 3 &	{[1,1,1,0,$\dots$,0,1,1]} & 4 \\
			{[1,1,1,0,$\dots$,0,1,1,1]} & 3&	{[1,1,1,0,$\dots$,0,1,1,1,1]} & 2\\
		\hline
	\end{tabular}
\end{table}

\begin{table}[t]
	\centering
	\setlength{\tabcolsep}{4pt}
	\renewcommand{\arraystretch}{1.3}
	\caption{Results for Example 1 for 21 samples of $x_0$.}
	\label{t:ex1}

	\begin{tabular}{|c|c|c|}
		\hline
		$T=20$ & std. B\&B  & node wght\\ 
%		&B\&B   & weighting  \\ 
		\hline
		avg. \# nodes & 1125.15 & 925.24  \\
		median \# nodes & 1126	& 890	\\
		avg. runtime (s) & 3000 &  466.01	\\
		median runtime & 3000 & 274.83\\
		best \# nodes & 1111 &  508 \\
		best runtime (s) & 3000.12 & 90.63\\
		avg. subopt.& 423 & 0\\
		\hline
	\end{tabular} 	\begin{tabular}{|c|c|c|}
		\hline
		$T=10$ & std. B\&B  & node wght\\
		%	& B\&B & weighting  \\ 
		\hline
		avg. \# nodes & 503.43 & 147.71 	 \\
		median \# nodes & 515	& 102	\\
		avg. runtime (s) & 739.87 & 149.95  	 \\
		median runtime & 357.19 & 7.01 \\
		best \# nodes & 135 &  34	\\
		best runtime (s) & 15.85 & 2.08   \\
		avg. subopt. & 0 & 0 \\
		\hline
	\end{tabular}	
\end{table}
A termination limit of 3000 seconds is set for each test, however, the simple standard branch-and-bound algorithm often failed to terminate within this time for many of the larger problems, which is the cause of some suboptimality. As evidenced by the results in Table \ref{t:ex1}, the proposed  node weighting, which exploits the turnpike property of MIOCP \eqref{eq:lit_ex}, yields solutions much more quickly.

\subsubsection*{Example 2}\label{sec:ex2}
As a second example we consider
\begin{align} \label{my_ex}
%\min_{x,u,z} & \sum_{k=1}^T A u(k) + B z(k) + C x(k) + D u(k)^2 + E z(k)^2 + F x(k)^2 \nonumber\\
\min_{z} &\quad \sum_{k=0}^{N-1} \ell(z(k)) \nonumber\\ %z A z^\top + B z^\top \nonumber\\
%&\min_{x,u,z} & \sum_{k=1}^T [c_0, c_1, c_2, c_3, c_4, c_5][u(k), z(k), x(k),  u(k)^2,  z(k)^2,  x(k)^2]^\top \nonumber\\
\text{s. t. } &\forall k \in \{0,\dots,N-1\},\\
%& x(k+1) = \begin{pmatrix} 0& I\\0& 0\end{pmatrix}x(k) + \begin{pmatrix} 0\\ \vdots\\0\\ \alpha\end{pmatrix}^\top u(k) + \begin{pmatrix} \beta_1\\ \vdots\\ \beta_{n-1}\\ \beta_n\end{pmatrix}^\top z(k),\nonumber
 x(k+1) &= Ex(k) + B_1 u(k) + B_2 v(k), ~x(0) =x_0\nonumber \\
 z(k)&\in \mbb{R}^{n_x} \times \{0,\,1\}\nonumber
\end{align}
with $E=\begin{bmatrix} 0& I\\0& 0\end{bmatrix}$m $I \in \mbb{R}^{n_x-1 \times n_x-1}$ is the identity matrix, $\ell(z) =  10u+ v(k) + 100u^2 + 100x^\top x$, and $B_1 = \begin{bmatrix}0&0&\dots&0&1\end{bmatrix}^\top$, $B_2 = \begin{bmatrix}1&1&\dots&1&1\end{bmatrix}^\top$. The turnpike is at $\bar z^\star = 0$.% 

Shown in Figure \ref{fig:results2} is an example of the optimal solutions for $N=20$ and three state variables $n_x=3$. 
Observe that while Problem \eqref{eq:lit_ex} exhibits a so-called leaving arc---i.e. the optimal solutions depart from the turnpike towards the end of the horizon---this is not the case in Problem \eqref{my_ex}. 

The intial guesses $\mcl{V}_0$ are constructed as partial guesses of the integer controls corresponding to $\overline{v}^\star$ for $k\geq \hat k$ with $\hat k=\{2, \dots,6\}$. Each time Algorithm~1 is called only one of the guesses is passed ($\#\mcl{V}_0 =1$), its weight is set to $w_{0,1} = 4\cdot \max_{w\in\mcl{W}} w$. 
Shown in Table \ref{t:ex2} are the aggregated results for each combination of $x_0=\begin{bmatrix}-0.9 &-0.8&\dots&0.8&0.9\end{bmatrix}^\top +r$  for $N = 10, 20$ and $40$,  where $r$ is a uniformly distributed random vector whose entries range between $-0.1$ and $0.1$. Overall we consider 19 different samples of $x_0$. Moreover, we consider the dimension of the state to be $n_x = 30$. 
As in Example \eqref{eq:lit_ex}, ``node wght" denotes the results from Algorithm~1 using this weighting, while ``std. B\&B" does not use any initial guesses.
The results seen in Table \ref{t:ex2} illustrate the dramatic benefit even a simply node weighting strategy can give. It can be seen that the standard branch-and-bound method is an order of magnitude slower in the smallest case, and this gap only increases as the length of the turnpike increases. Note that all algorithms converge to the optimal solution. Part of the reason for the quick convergence is that rearranging nodes in the setup of Algorithm~1 results in a decision tree with some infeasible or integer-feasible solutions at the first nodes that are explored. This prunes many of the subsequent nodes and greatly reduces the search space.

\begin{figure}[t]
		\centering
	\includegraphics[width=0.239\textwidth]{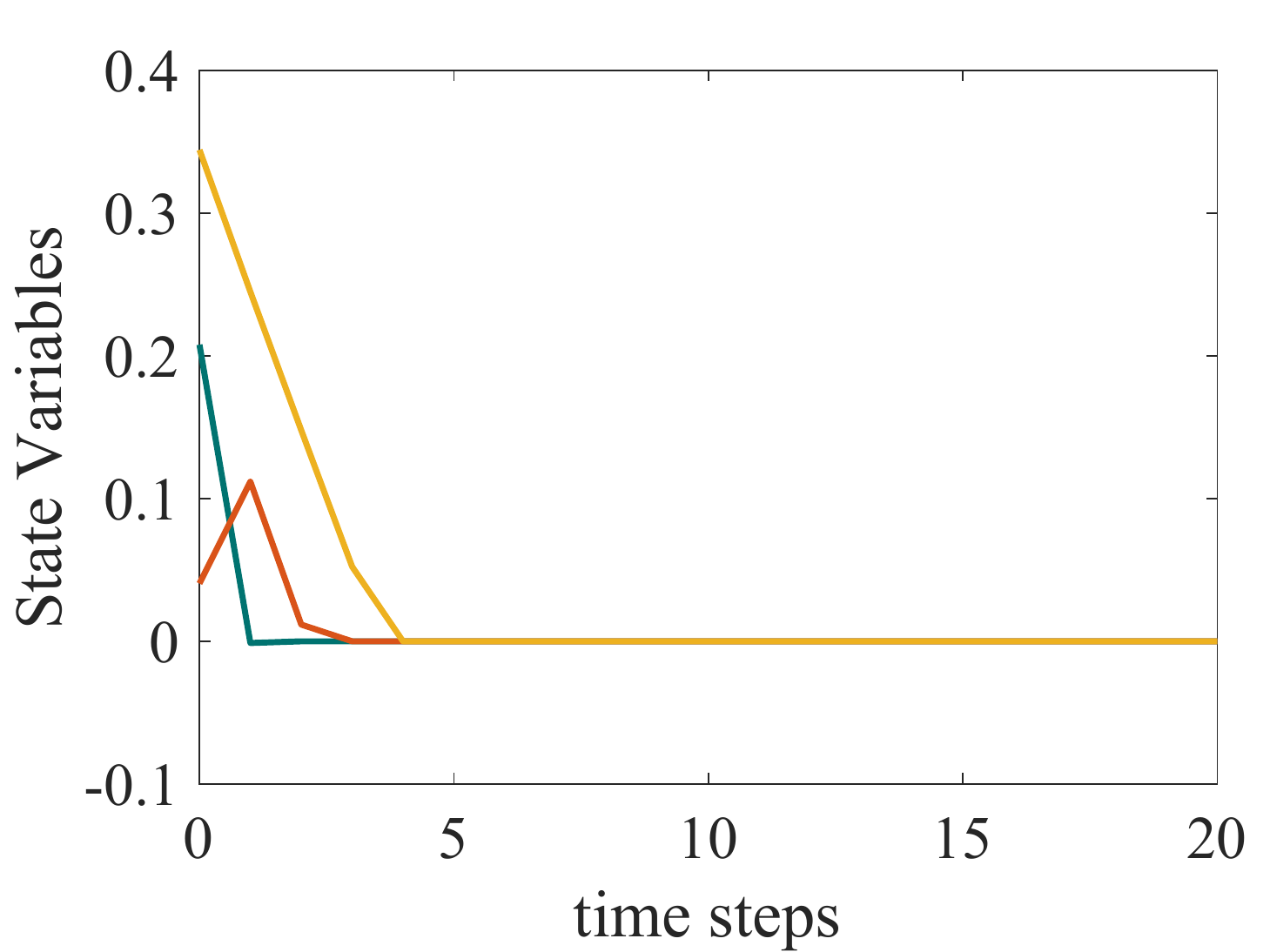}	\includegraphics[width=0.239\textwidth]{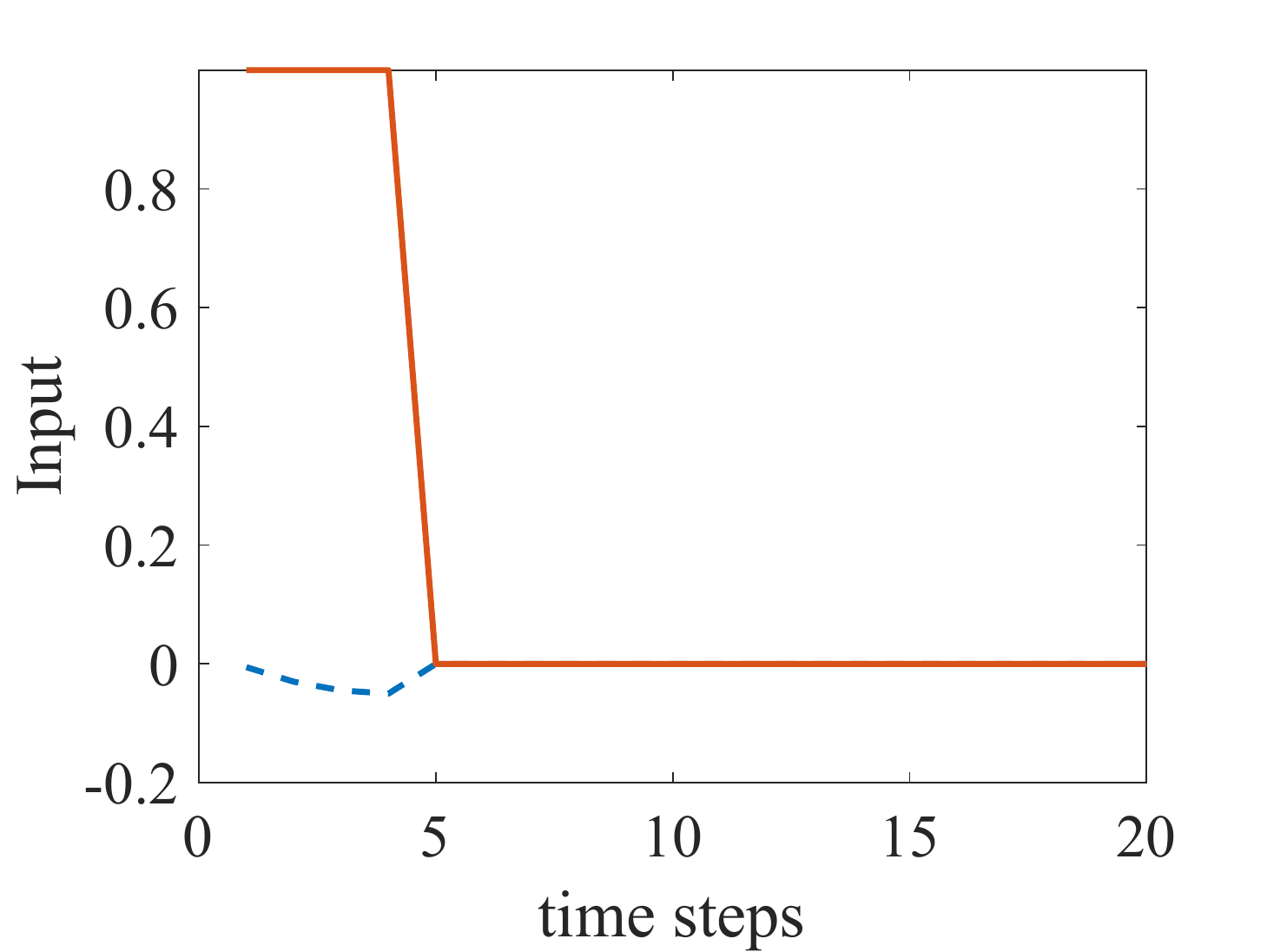}
	\caption{Optimal solution for Example~2 with $N=20$ and $\dim(x)=n_x=3$. Left: state trajectories. Right: controls $u$ (dashed blue) $v$ (solid red).}
	\label{fig:results2}
\end{figure}
\begin{table}[t]	
	\centering
	\setlength{\tabcolsep}{4pt}
	\renewcommand{\arraystretch}{1.3}
	\caption{Results for Example~2 for 19 samples of $x_0\in\mbb{R}^{30}$ and five different guesses $\mcl{V}_0 =\{\mbf{V}_{0,1}\}$.}
	\label{t:ex2}
	\begin{tabular}{|c|c|c|}
		\hline
		$T=10$ &  std. & node  \\ 
		$n_x=30$&  B\&B & wght\\ 
		\hline
		avg. \# nodes  & 51.79 & 6.11  \\
		median \# nodes  & 52	& 6\\
		avg. runtime (s) & 2.96 &  0.32 \\
		median time (s) & 2.46& 0.31\\
		best \# nodes & 36 & 2  \\
		best runtime (s) & 1.61 & 0.11 
		 \\
%		avg. subopt. & 0 &  0\\
		\hline
	\end{tabular}
%	\begin{tabular}{|c|c|c|}
%		\hline
%		$T=20$ & std.  & node  \\ 
%		$n_x=30$& B\&B & wght  \\ 
%		\hline
%		avg. \# nodes  & 110.84 & 6.11 \\
%		median \# nodes  & 112& 6 \\
%		avg. runtime (s) & 9.39 & 0.46 	\\
%		median time (s) &9.57 & 0.46\\
%		best \# nodes & 78 & 2   \\
%		best runtime (s) & 4.81 & 0.15  \\
%		%avg. subopt. & 0 & 0\\
%		\hline
%	\end{tabular}
	\begin{tabular}{|c|c|c|}
		\hline
		$T=40$ & std.  & node  \\ 
		$n_x=30$& B\&B & wght  \\ 
		\hline
		avg. \# nodes  & 228.5 & 6.42 \\
		median \# nodes  & 232& 6 \\
		avg. runtime (s) & 35.7 & 0.77 	\\
		median time (s) &36.7 & 0.73\\
		best \# nodes & 158 & 2  \\
		best runtime (s) & 16.31 & 0.23 \\
%		avg. subopt. & 0 & 0\\
		\hline
	\end{tabular}
\end{table}

%%%%%%%%%%%%%%%%%%%%%%%%%%%%%%%%%%%%%%%%%%%%
\section{Conclusions and Outlook}\label{sec:conc}\vspace*{-2mm}

This note has taken first steps towards a turnpike theory for  mixed-integer OCPs, thus paving the road for a better understanding of properties of parametric MIOCPs. Specifically, we have provided sufficient turnpike conditions based on a dissipativity notion of MIOCPs.  We have also shown that the discrete controls will enter the turnpike exactly, while for the continuous controls this is not necessarily the case. For the special case of linear-quadratic MIOCPs we have presented an easy to check sufficient condition for dissipativity of MIOCPs on compact sets. Moreover, we have discussed how these insights can be easily leveraged to design node-weighted branch-and-bound schemes.  While the present work appears to be the very first to discuss turnpikes in MIOCPs, at this stage our numerical results are an initial step relying on prototypical implementations. 
Future work will focus on several aspects including further exploitation of the turnpike phenomenon in branch-and-bound schemes for convex and non-convex MIOCPs.

\subsubsection*{Acknowledgement}The authors acknowledge helpful   comments of Veit Hagenmeyer.

\bibliographystyle{plain}

%\appendix

%\appendix

%\renewcommand*{\bibfont}{\scriptsize}
 \bibliographystyle{plain}
 {\scriptsize
\bibliography{MyBibliography,literature_latin1} 
}

\end{document}